\documentclass[journal]{IEEEtran}

\usepackage[cmex10]{amsmath}  
\usepackage{amssymb}
\usepackage{graphicx,color}   

\usepackage{ifxetex}
\ifxetex
\usepackage{xunicode}
\else
\usepackage[utf8]{inputenc}
\usepackage[T1]{fontenc}
\usepackage{microtype}
\fi

\usepackage{cite}  

\newtheorem{theorem}{Theorem}
\newtheorem{remark}{Remark}
\newtheorem{lemma}{Lemma}

\newtheorem{proposition}{Proposition}
\newtheorem{conjecture}{Conjecture}
\newtheorem{corollary}{Corollary}

\begin{document}

\title{On the Equal-Rate Capacity of the\\ AWGN Multiway Relay Channel}
\author{ Lawrence Ong, Christopher M.\ Kellett, and Sarah J.\ Johnson
\thanks{ Part of the material in this paper was presented at the 2010 IEEE International Symposium on Information Theory, Austin, USA, June 2010. }
\thanks{The authors are with the School of Electrical Engineering and
Computer Science, The University of Newcastle, Callaghan, NSW 2308,
Australia (email: lawrence.ong@cantab.net; \{chris.kellett,sarah.johnson\}@newcastle.edu.au).}
\thanks{This work is supported by the Australian Research Council under grants DP0877258 and DP1093114.} }

\maketitle

\begin{abstract}
The $L$-user additive white Gaussian noise multiway relay channel is investigated, where $L$ users exchange information at the same rate through a single relay. A new achievable rate region, based on the functional-decode-forward coding strategy, is derived. For the case where there are three or more users, and all nodes transmit at the same power, the capacity is obtained. For the case where the relay power scales with the number of users, it is shown that both compress-forward and functional-decode-forward achieve rates within a constant number of bits of the capacity at all SNR levels; in addition, functional-decode-forward outperforms compress-forward and complete-decode-forward at high SNR levels.
\end{abstract}



\section{Introduction}

We study the additive white Gaussian noise (AWGN) multiway relay channel (MWRC), in which $L$ users exchange full information at the same/equal rate via a relay. This channel models centralized networks where users communicate among themselves via a base station, e.g., multiuser mobile video conference calls.

\subsection{Main Results}
We consider the AWGN MWRC with $L$ users and one relay where is no direct link between any pair of users, as depicted in Fig.~\ref{fig:mwrc}. Each user has a transmitted power constraint $P$, the relay has a transmitted power constraint $P_0$, and the receiver noise at all users and the relay is normalized to one. Each user is to transmit its data to all other users at the same rate. We obtain the capacity when
\begin{equation}
P_0 \leq \max \left\{ \frac{LP-1}{2},(LP+1)^{\frac{L-1}{L}} -1 \right\}.
\end{equation}
More specifically, we have
\begin{enumerate}
\item the capacity if $L \geq 3$ and $P_0 = P$, i.e., when there are three or more users, and when all users and the relay transmit at the same power;
\item achievable rates within a constant $\frac{1}{2(L-1)\ln 2}$ bits of the capacity for all $P$ if $P_0 = LP$, i.e., when relay power scales with the number of users.
\end{enumerate}

\begin{figure}[t]
\centering
\resizebox{8cm}{!}{ 
\begin{picture}(0,0)%
\includegraphics{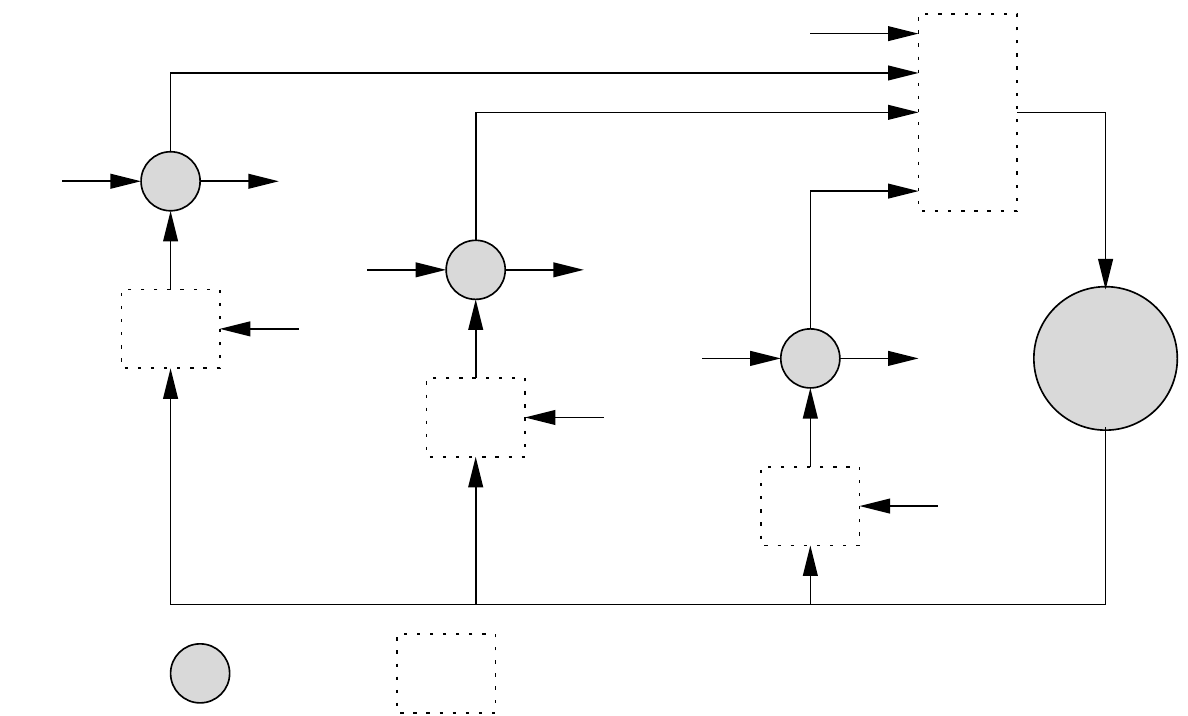}%
\end{picture}%
\setlength{\unitlength}{4144sp}%
\begingroup\makeatletter\ifx\SetFigFont\undefined%
\gdef\SetFigFont#1#2#3#4#5{%
  \fontsize{#1}{#2pt}%
  \fontfamily{#3}\fontseries{#4}\fontshape{#5}%
  \selectfont}%
\fi\endgroup%
\begin{picture}(5414,3255)(301,-2908)
\put(5101,-1421){\makebox(0,0)[lb]{\smash{{\SetFigFont{12}{14.4}{\familydefault}{\mddefault}{\updefault}{\color[rgb]{0,0,0}(relay)}%
}}}}
\put(991,-1186){\makebox(0,0)[lb]{\smash{{\SetFigFont{12}{14.4}{\familydefault}{\mddefault}{\updefault}{\color[rgb]{0,0,0}$\sum$}%
}}}}
\put(2251,-2761){\makebox(0,0)[lb]{\smash{{\SetFigFont{12}{14.4}{\familydefault}{\mddefault}{\updefault}{\color[rgb]{0,0,0}$\sum$}%
}}}}
\put(1126,-151){\makebox(0,0)[lb]{\smash{{\SetFigFont{12}{14.4}{\familydefault}{\mddefault}{\updefault}{\color[rgb]{0,0,0}$X_1$}%
}}}}
\put(4141,-466){\rotatebox{90.0}{\makebox(0,0)[lb]{\smash{{\SetFigFont{12}{14.4}{\familydefault}{\mddefault}{\updefault}{\color[rgb]{0,0,0}$\dotsm$}%
}}}}}
\put(3781,164){\makebox(0,0)[lb]{\smash{{\SetFigFont{12}{14.4}{\familydefault}{\mddefault}{\updefault}{\color[rgb]{0,0,0}$N_0$}%
}}}}
\put(4636,-151){\makebox(0,0)[lb]{\smash{{\SetFigFont{12}{14.4}{\familydefault}{\mddefault}{\updefault}{\color[rgb]{0,0,0}$\sum$}%
}}}}
\put(811,-871){\makebox(0,0)[lb]{\smash{{\SetFigFont{12}{14.4}{\familydefault}{\mddefault}{\updefault}{\color[rgb]{0,0,0}$Y_1$}%
}}}}
\put(4051,-691){\makebox(0,0)[lb]{\smash{{\SetFigFont{12}{14.4}{\familydefault}{\mddefault}{\updefault}{\color[rgb]{0,0,0}$X_L$}%
}}}}
\put(3736,-1681){\makebox(0,0)[lb]{\smash{{\SetFigFont{12}{14.4}{\familydefault}{\mddefault}{\updefault}{\color[rgb]{0,0,0}$Y_L$}%
}}}}
\put(4591,-1996){\makebox(0,0)[lb]{\smash{{\SetFigFont{12}{14.4}{\familydefault}{\mddefault}{\updefault}{\color[rgb]{0,0,0}$N_L$}%
}}}}
\put(2206,-1276){\makebox(0,0)[lb]{\smash{{\SetFigFont{12}{14.4}{\familydefault}{\mddefault}{\updefault}{\color[rgb]{0,0,0}$Y_2$}%
}}}}
\put(2521,-331){\makebox(0,0)[lb]{\smash{{\SetFigFont{12}{14.4}{\familydefault}{\mddefault}{\updefault}{\color[rgb]{0,0,0}$X_2$}%
}}}}
\put(1666,-1186){\makebox(0,0)[lb]{\smash{{\SetFigFont{12}{14.4}{\familydefault}{\mddefault}{\updefault}{\color[rgb]{0,0,0}$N_1$}%
}}}}
\put(5041,-331){\makebox(0,0)[lb]{\smash{{\SetFigFont{12}{14.4}{\familydefault}{\mddefault}{\updefault}{\color[rgb]{0,0,0}$Y_0$}%
}}}}
\put(5041,-2311){\makebox(0,0)[lb]{\smash{{\SetFigFont{12}{14.4}{\familydefault}{\mddefault}{\updefault}{\color[rgb]{0,0,0}$X_0$}%
}}}}
\put(3061,-1996){\makebox(0,0)[lb]{\smash{{\SetFigFont{12}{14.4}{\familydefault}{\mddefault}{\updefault}{\color[rgb]{0,0,0}$\dotsm$}%
}}}}
\put(3061,-1591){\makebox(0,0)[lb]{\smash{{\SetFigFont{12}{14.4}{\familydefault}{\mddefault}{\updefault}{\color[rgb]{0,0,0}$N_2$}%
}}}}
\put(316,-511){\makebox(0,0)[lb]{\smash{{\SetFigFont{12}{14.4}{\familydefault}{\mddefault}{\updefault}{\color[rgb]{0,0,0}$W_1$}%
}}}}
\put(1621,-511){\makebox(0,0)[lb]{\smash{{\SetFigFont{12}{14.4}{\familydefault}{\mddefault}{\updefault}{\color[rgb]{0,0,0}$\underline{\hat{W}_1}$}%
}}}}
\put(1711,-916){\makebox(0,0)[lb]{\smash{{\SetFigFont{12}{14.4}{\familydefault}{\mddefault}{\updefault}{\color[rgb]{0,0,0}$W_2$}%
}}}}
\put(2431,-916){\makebox(0,0)[lb]{\smash{{\SetFigFont{12}{14.4}{\familydefault}{\mddefault}{\updefault}{\color[rgb]{0,0,0}$2$}%
}}}}
\put(2386,-1591){\makebox(0,0)[lb]{\smash{{\SetFigFont{12}{14.4}{\familydefault}{\mddefault}{\updefault}{\color[rgb]{0,0,0}$\sum$}%
}}}}
\put(3241,-1321){\makebox(0,0)[lb]{\smash{{\SetFigFont{12}{14.4}{\familydefault}{\mddefault}{\updefault}{\color[rgb]{0,0,0}$W_L$}%
}}}}
\put(3961,-1321){\makebox(0,0)[lb]{\smash{{\SetFigFont{12}{14.4}{\familydefault}{\mddefault}{\updefault}{\color[rgb]{0,0,0}$L$}%
}}}}
\put(3916,-1996){\makebox(0,0)[lb]{\smash{{\SetFigFont{12}{14.4}{\familydefault}{\mddefault}{\updefault}{\color[rgb]{0,0,0}$\sum$}%
}}}}
\put(3016,-916){\makebox(0,0)[lb]{\smash{{\SetFigFont{12}{14.4}{\familydefault}{\mddefault}{\updefault}{\color[rgb]{0,0,0}$\underline{\hat{W}_2}$}%
}}}}
\put(1036,-511){\makebox(0,0)[lb]{\smash{{\SetFigFont{12}{14.4}{\familydefault}{\mddefault}{\updefault}{\color[rgb]{0,0,0}$1$}%
}}}}
\put(2656,-2761){\makebox(0,0)[lb]{\smash{{\SetFigFont{12}{14.4}{\familydefault}{\mddefault}{\updefault}{\color[rgb]{0,0,0}real addition}%
}}}}
\put(1441,-2761){\makebox(0,0)[lb]{\smash{{\SetFigFont{12}{14.4}{\familydefault}{\mddefault}{\updefault}{\color[rgb]{0,0,0}node}%
}}}}
\put(5311,-1231){\makebox(0,0)[lb]{\smash{{\SetFigFont{12}{14.4}{\familydefault}{\mddefault}{\updefault}{\color[rgb]{0,0,0}$0$}%
}}}}
\put(4546,-1321){\makebox(0,0)[lb]{\smash{{\SetFigFont{12}{14.4}{\familydefault}{\mddefault}{\updefault}{\color[rgb]{0,0,0}$\underline{\hat{W}_L}$}%
}}}}
\end{picture}%
}
\caption{An $L$-user AWGN MWRC, where $\underline{\hat{W}_i}$ is user $i$'s estimate of all other users' messages}
\label{fig:mwrc}
\end{figure}
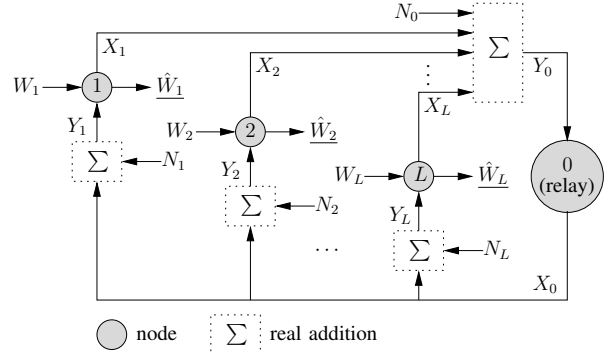

\subsection{Related Work}

The MWRC is an extension of the two-way relay channel (TWRC), i.e., $L=2$. Since there is an embedded relay channel from user 1 to user 2 via the relay, and vice versa, coding strategies proposed for relay networks~\cite{covergamal79,kramergastpar04}, i.e., complete-decode-forward,\footnote{This coding strategy is commonly known as decode-forward. We modify the name to distinguish this coding strategy and the functional-decode-forward coding strategy to be discussed next.} compress-forward, and amplify-forward, have been modified for the TWRC~\cite{rankovwittneben05,rankovwittneben06}.

Besides the aforementioned classical relaying strategies for relay networks,
Yu et al.~\cite{wuchoukung05} and Larsson et al.~\cite{larssonjohanssonsunell05} used the idea of network coding~\cite{ahlswedecai00} for the TWRC and proposed that the relay performs a modulo-two summation of the users' messages, and broadcasts the summation. After getting the sum of the messages, each user can then obtain the other user's message by subtracting its own message from the sum. The network-coding-like operation can also be achieved at the modulation level, if amplify-forward is carried out using binary phase-shift keying (BPSK)~\cite{zhangliew06} or minimum-shift keying (MSK) modulation~\cite{kattigollakota07}.

Since the relay needs to broadcast only the modulo-sum of the messages (or in general a function defined such that, having the function and its message, each user is able to decode the message of the other user), the relay may directly decode the function that it will broadcast without needing to decode the individual messages. This strategy is known as functional-decode-forward and also compute-and-forward. This strategy was studied on the binary TWRC using binary linear codes~\cite{namchung08}, and on the AWGN TWRC using lattice codes~\cite{namchung08,wilsonnarayananpfisersprintson10,namchunglee09}. The scenario of having a relay decoding a linear function of messages from multiple senders was studied by Nazer and Gastpar~\cite{nazergastpar08eu,nazergastpar11}.

Complete-decode-forward, compress-forward, and amplify-forward for the TWRC were subsequently extended to the MWRC by G\"und\"uz et al.~\cite{gunduzyener09}. We extended functional-decode-forward to the binary MWRC~\cite{ongjohnsonkellett10cl}. The idea was to split the uplink into $(L-1)$ blocks, and when the user pairs $(1,2), (2,3), \dotsc, (L-1,L)$ transmit in the blocks respectively. In each block, the relay decodes a function (in this case, the modulo-two sum) of the two corresponding messages, and broadcasts the function back to the users. If each user can decode the $(L-1)$ functions from the relay, it can then decode the messages of all the other users.  In this paper, we extend this strategy to the AWGN MWRC. We propose a {\em rotated} scheme such that for each message to be sent, the uplink is again split into $(L-1)$ blocks, but different pairs of users are chosen from $\{ (1,2), (2,3), \dotsc, (L-1,L), (L,1)\}$ for each message in a round-robin fashion. This modification enables each user to transmit at a higher power when active while keeping the average transmitted power to $P$. For each block, the two active users transmit using lattice codes as they do in the AWGN TWRC~\cite{wilsonnarayananpfisersprintson10}.

While the aforementioned functional-decode-forward for the binary MWRC~\cite{ongjohnsonkellett10cl} and for the AWGN MWRC (this paper) deals with users transmitting at the same rate, we have recently extended this coding strategy to the finite field MWRC where the users transmit at possibly different rates\cite{ongmjohnsonit11}. Note, however, that the ``rotation'' of the functional-decode-forward scheme introduced in this paper is not required for the binary MWRC or the finite field MWRC.

The {\em restricted} MWRC was studied by G\"und\"uz et al.~\cite{gunduzyener09} and Ong et al.~\cite{ongjohnsonkellett10cl}, where each user's transmitted signals can only depend on it message; however, the {\em unrestricted} MWRC is studied in this paper where each user's transmitted signals can depend on both its message and its previously received signals.

It was shown by G\"und\"uz et al.~\cite{gunduzyener09} that for the restricted AWGN MWRC, complete-decode-forward performs poorly at high SNR, compress-forward achieves rates within a constant fraction of a bit of the capacity at all SNRs, and amplify-forward is always worse than compress-forward at all SNRs. For the case where all nodes transmit at the same power, the same authors noted that complete-decode-forward achieves the capacity with ``a smaller number of users''.

Note that the capacity of the restricted AWGN MWRC is necessarily upper bounded by that of the unrestricted AWGN MWRC.
For the unrestricted AWGN MWRC with three or more users, when all nodes transmit at the same power, we obtain the capacity by showing that complete-decode-forward achieves the capacity for SNR $\leq 0$ dB, and functional-decode-forward achieves the capacity for SNR $\geq 0$ dB. As we shall see, neither coding strategy utilizes feedback at the users, and hence we incidentally obtain the capacity for the corresponding restricted AWGN MWRC (for three or more users and equal transmitted power). 

For the case where the relay's transmitted power scales with the number of users, we show that none of the three coding strategies can achieve the capacity upper bound. However, we show that both compress-forward and functional-decode-forward achieve within a constant fraction of a bit of the capacity. At high SNR, we also show that functional-decode-forward achieves higher rates than complete-decode-forward and compress-forward. Furthermore, numerical results suggest that for any $L$ and at any SNR, either complete-decode-forward or functional-decode-forward always outperforms compress-forward.

\subsection{Organization}
The rest of the paper is organized as follows: In Section~\ref{section:channel-model} we define the AWGN MWRC considered in this paper. We derive an upper bound to the capacity in Section~\ref{section:upper-bound}, and review two existing lower bounds to the capacity and derive a new lower bound based on functional-decode-forward in Section~\ref{section:existing}. We present capacity results for the AWGN MWRC in Section~\ref{section:capacity}. Section~\ref{section:conclusion} concludes the paper.

\section{Channel Model}\label{section:channel-model}
Fig.~\ref{fig:mwrc} depicts the $L$-user AWGN MWRC considered in this paper, where
the uplink and the downlink channels are separated, i.e., there are no direct user-to-user links. Nodes $1,2,\dotsc,L$ are the users (where $L \geq 2$), and node $0$ is the relay. We denote by $X_i$ node $i$'s input to the channel, $Y_i$ the channel output received by node $i$, and $W_i$ node $i$'s message. We consider full data exchange, i.e., each user is to decode the messages from all other users.

The AWGN MWRC is defined as follows:
\begin{enumerate}
\item The uplink channel is the sum of all users' channel inputs and the relay's receiver noise
\begin{equation}
Y_0 = \sum_{i=1}^L X_i + N_0,
\end{equation}
where $N_0$ is an i.i.d. zero-mean Gaussian random variable with variance $\sigma_0^2$.
\item The downlink consists of an independent point-to-point AWGN channel for each user, $i \in \{1,2,\dotsc,L\}$
\begin{equation}
Y_i = X_0 + N_i,
\end{equation}
where $N_i$ is user $i$'s receiver noise and is an i.i.d. zero-mean Gaussian random variable with variance $\sigma_i^2$.
\end{enumerate}

Consider $n$ uses of the AWGN MWRC.  We use bold letters to define vectors of length $n$ (in time), e.g., $\boldsymbol{X} = (X[1], X[2], \dotsc, X[n])$, where $X[t]$ is the random variable $X$ at time $t$. We also define $X_{\mathcal{S}} = \{X_i: i \in \mathcal{S}\}$. We impose the following average transmitted power constraints on the users and the relay:  $\sum_{t=1}^n E[X_j ^2[t]]/n \leq P_j$, for all $j \in \{0,1,\dotsc, L\}$.

Unlike previous work on the MWRC~\cite{gunduzyener09,ongjohnsonkellett10cl}, we consider the \emph{unrestricted} MWRC in the sense that the transmitted signals of each user can depend on both its message and its previously received signals. We consider a block code of $n$ channel uses consisting of the following functions:
\begin{enumerate}
\item a set of encoding functions at each node $j$: $X_j[t] = f_{j,t} (W_j, Y_j[1], Y_j[2], \dotsc, Y_j[t-1])$, for all $j \in \{0,1, \dotsc, L\}$ and $t \in \{1,2,\dotsc,n\}$, where $W_0 = \varnothing$; 
\item a decoding function at each user $i$: $\underline{\hat{W}_i} \triangleq(\hat{W}_{i,1}, \dotsc, \hat{W}_{i,i-1}, \hat{W}_{i,i+1}, \dotsc, \hat{W}_{i,L}) = g_i(\boldsymbol{Y}_i,W_i)$, for all  $i \in \{1,2,\dotsc, L\}$, where $\hat{W}_{i,j}$ is user $i$'s estimate of $W_j$.
\end{enumerate}

Assuming that the users' messages are independent and each message $W_i$ is uniformly distributed in $\{1,2,\dotsc, 2^{nR_i}\}$, the average probability of error is defined as $P_\text{e} \triangleq \Pr \Big\{\hat{W}_{i,j} \neq W_j, \text{ for some } j \in \{1,2,\dotsc, L\}$ $\text{and some } i \neq j\Big\}$.
A rate tuple $(R_1,R_2, \dotsc, R_L)$ is said to be \emph{achievable} if and only if (iff) the following is true: For any $\epsilon > 0$, there exists for sufficiently large $n$ at least one block code defined above such that $P_\text{e} < \epsilon$.

We say that a node can \emph{reliably} decode a message iff the average probability that the node wrongly decodes the message can be made arbitrarily small.

In this paper, we consider symmetrical networks, i.e.,  $P_i=P$ for all users $i \in \{1,2,\dotsc,L\}$, and $\sigma_j^2=1$ for all nodes $j \in \{0,1,\dotsc,L\}$. So, in a symmetrical AWGN MWRC, SNRs at all the users are equal, given by $\frac{P_0}{\sigma_i^2} = P_0$, and the SNR at the relay is $\frac{P}{\sigma_0^2} = P$.
Furthermore, we focus on the \emph{equal rate} $R = R_i$, $\forall i \in \{1,2,\dotsc,L\}$. We say that the equal rate $R$ is achievable iff the rate tuple $(R,R,\dotsc,R)$ is achievable.

We define the \emph{equal-rate capacity} of the MWRC (also known as the symmetrical capacity~\cite{gunduzyener09}) as $C \triangleq \sup \{R: (R,R,\dotsc,R) \text{ is achievable}\}$.
The equal rate is useful in systems where all users have the same amount of information to send, or in \emph{fair} systems where every user is to be given the same guaranteed uplink \emph{bandwidth}, i.e., each user can send data up to a certain rate, at which all other users are able to decode.

With a slight abuse of terminology, we say that a coding strategy ``achieves'' the equal rate $R^*$, if the strategy achieves all equal rates $R < R^*$. So, a coding strategy is said to achieve the equal-rate capacity if it achieves all equal rates $R < C$.

\section{An Upper Bound to The Equal-Rate Capacity}\label{section:upper-bound}

In this section, we derive an upper bound on the equal-rate capacity of the AWGN MWRC.
First, we quote a result by Cover and Thomas~\cite[page 589]{coverthomas06}. Consider an arbitrary network of $m$ nodes where node $i$ sends information  to node $j$ at the rate $R_{i,j}$. We have the following lemma:
\begin{lemma}\label{lemma:cover-cut-set}
If $\{R_{i,j}\}$  are achievable, there exists some joint probability distribution $p(x_1,x_2,\dotsc,x_m)$ such that
\begin{equation}
\sum_{i \in \mathcal{S}, j \in \mathcal{S}^c} R_{i,j} \leq I(X_\mathcal{S};Y_{\mathcal{S}^c} | X_{\mathcal{S}^c} ),\label{eq:thm-ub-new}
\end{equation}
for all $\mathcal{S} \subset \{1,2,\dotsc,m\}$, and $\mathcal{S}^c = \{1,2,\dotsc,m\} \setminus \mathcal{S}$.
\end{lemma}

Upper bounds of the type in Lemma~\ref{lemma:cover-cut-set} are often called cut-set upper bounds. A cut-set upper bound to the capacity of a network is the maximum rate that information can be transferred across a \emph{cut} separating two disjoint sets of nodes, assuming that all nodes on each side of the cut can fully cooperate.

By applying Lemma~\ref{lemma:cover-cut-set} to the AWGN MWRC, we have the following:
\begin{lemma}\label{lemma:binary-upper-bound}
The equal-rate capacity of the AWGN MWRC is upper-bounded by
\begin{align}
C &\leq \min_{\ell \in \{1,2,\dotsc,L-1\}} \left\{ \frac{1}{2\ell} \log ( 1 + \ell^2 P ), \frac{1}{2(L-1)} \log ( 1 + P_0) \right\} \nonumber \\
&\triangleq R_\text{UB}. \label{eq:upper-bound}
\end{align}
\end{lemma}

\begin{proof}
Let a strict, non-empty subset of users be $\mathcal{U} \subset \{1,2,\dotsc, L\}$, and define $\mathcal{U}^\text{c} = \{1,2,\dotsc, L\} \setminus \mathcal{U}$, i.e., the set of users not in $\mathcal{U}$. We have $1 \leq |\mathcal{U}| \leq L-1$. Consider the cut separating $\mathcal{U}$ and $\{0\} \cup \,  \mathcal{U}^\text{c}$, i.e., the relay is grouped with $\mathcal{U}^\text{c}$.
The total information flow from $\mathcal{U}$ to $\{0\} \cup \,  \mathcal{U}^\text{c}$ is $W_\mathcal{U}$ with the sum rate of $|\mathcal{U}|R$. We have the following rate constraint on these messages, for all $\mathcal{U} \subset \{1,2,\dotsc, L\}$:
\begin{subequations}
\begin{align}
|\mathcal{U}|R & \leq I(X_\mathcal{U}; Y_{\{0\} \cup \,  \mathcal{U}^\text{c}} | X_{\{0\} \cup \,  \mathcal{U}^\text{c}} )\\
 &= h(Y_{\{0\} \cup \,  \mathcal{U}^\text{c}}|X_{\{0\} \cup \,  \mathcal{U}^\text{c}}) - h(Y_{\{0\} \cup \,  \mathcal{U}^\text{c}}|X_\mathcal{U},X_{\{0\} \cup \,  \mathcal{U}^\text{c}} )\\
&= h(Y_0,Y_{\mathcal{U}^\text{c}}|X_0,X_{\mathcal{U}^\text{c}} ) - h(Y_0,Y_{\mathcal{U}^\text{c}}|X_{\{0,1,\dotsc,L\}}) \\
&\leq h\left( \sum_{i\in\mathcal{U}} X_i + N_0,N_{\mathcal{U}^\text{c}} \right) - h(N_0,N_{\mathcal{U}^\text{c}})\\
&= h\left( \sum_{i\in\mathcal{U}} X_i + N_0\right) + h(N_{\mathcal{U}^\text{c}}) - h(N_0) - h(N_{\mathcal{U}^\text{c}}) \label{eq:indep-noise}\\
&\leq \frac{1}{2} \log 2 \pi e \left( \left(\sum_{i \in \mathcal{U}} \sqrt{P_i}\right)^2 + \sigma_0^2 \right) - \frac{1}{2} \log 2 \pi e \sigma_0^2 \label{eq:gaussian-input}\\
R & \leq \frac{1}{2|\mathcal{U}|} \log \left( 1 + |\mathcal{U}|^2 P \right), \label{eq:all-u}
\end{align}
\end{subequations}
where \eqref{eq:indep-noise} is because $\left(\sum_{i\in\mathcal{U}} X_i + N_0\right)$ and $N_{\mathcal{U}^\text{c}}$ are independent, and so are all $N_i$ for $ i \in \{0,1,\dotsc,L\}$; \eqref{eq:gaussian-input} follows from the result that Gaussian inputs $\{X_i:i\in\mathcal{S}\}$ maximize the entropy subject to a second moment constraint~\cite[Theorem 8.6.5]{coverthomas06}, and by taking $X_i = X_j, \forall i,j \in \mathcal{U}$; and
\eqref{eq:all-u} follows from our definition that $\sigma_0^2 = 1$ and $P_i = P$, for all $i \in \{1,2,\dotsc, L\}$.



Now, we consider the  cut separating $\{0\} \cup \,  \mathcal{U}$ and $\mathcal{U}^\text{c}$. The total information flow from $\{0\} \cup \,  \mathcal{U}$ to $\mathcal{U}^\text{c}$ is again $W_\mathcal{U}$ with the sum rate  of $|\mathcal{U}|R$. We have the following rate constraint on these messages, for all  $\mathcal{U} \subset \{1,2,\dotsc, L\}$:
\begin{subequations}
\begin{align}
|\mathcal{U}|R
&\leq I(X_{\{0\} \cup \,  \mathcal{U}}; Y_{\mathcal{U}^c} | X_{\mathcal{U}^c} )\\
&= h(Y_{\mathcal{U}^c} | X_{\mathcal{U}^c} ) - h(Y_{\mathcal{U}^c}|X_{\{0,1,\dotsc,L\}})\\
&\leq h(Y_{\mathcal{U}^c}) - h(N_{\mathcal{U}^c})\\
&\leq h(Y_{\mathcal{U}^c}) - h(N_i), \quad \text{for some } i \in \mathcal{U}^\text{c}\\
R &\leq \frac{1}{|\mathcal{U}|} \left( h(Y_{\mathcal{U}^c}) - \frac{1}{2} \log 2 \pi e \right), \label{eq:common-sigma}
\end{align}
\end{subequations}
where \eqref{eq:common-sigma} is because $\sigma_i^2 = 1, \forall i$. The above must hold when $|\mathcal{U}| = L-1$, meaning that $\mathcal{U}^c$ is a singleton. So, we must have
\begin{subequations}
\begin{align}
R &\leq \frac{1}{|L-1|} \left( \frac{1}{2}\log 2 \pi e (1 + P_0 ) - \frac{1}{2} \log 2 \pi e \right) \label{eq:relay-gaussian}\\
&\leq \frac{1}{2(L-1)} \log ( 1 + P_0), \label{eq:ub-3}
\end{align}
\end{subequations}
where \eqref{eq:relay-gaussian} is because Gaussian input $X_0$ maximizes the entropy subject to a second moment constraint. 

Since the rate $R$ must be bounded by constraints  \eqref{eq:all-u} for all $1 \leq |\mathcal{U}| \leq L-1$ and  \eqref{eq:ub-3}, we have Lemma~\ref{lemma:binary-upper-bound}.
\end{proof}


\section{Lower Bounds to the Equal-Rate Capacity}\label{section:existing}

\subsection{Existing Strategies}

G\"und\"uz et al.~\cite{gunduzyener09} considered the AWGN MWRC where users, grouped into multiple {\em clusters}, exchange information through one relay. Users in each cluster fully exchange their messages, but they do not exchange any information across clusters. Setting the number of clusters to one, we get the MWRC considered in this paper. Hence, we have the following two achievable rates:

\subsubsection{Complete-Decode-Forward}
Using the complete-decode-forward coding strategy, the relay decodes all users' messages on the uplink. It then re-encodes and broadcasts a function of the messages back to the users on the downlink~\cite{rankovwittneben06,knopp06,gunduzyener09}. The function is also constructed such that each user can decode all other users' messages from the function and its own message. Complete-decode-forward can achieve the following rate~\cite[Proposition 3]{gunduzyener09}:
\begin{proposition}\label{theorem:cdf}
Consider an AWGN MWRC. Using the complete-decode-forward coding strategy, the following equal rates are achievable:
\begin{equation}
R < \min \left\{ \frac{1}{2L} \log ( 1 + LP), \frac{1}{2(L-1)} \log ( 1 + P_0) \right\} \triangleq R_\text{CDF}. \label{eq:cdf}
\end{equation}
\end{proposition}

\subsubsection{Compress-Forward}
Using the compress-forward coding strategy, the relay quantizes its received signals. It then encodes and broadcasts the quantized signals to the users~\cite{rankovwittneben06,schnurroechtering07,gunduzyener09}. The quantization level is determined such that each user can decode the quantized signals, which contain the sum of all users' transmission, uplink channel noise, and quantization noise. Subtracting its own transmission from the quantized signals, each user then decodes other users' messages. Compress-forward can achieve the following rate~\cite[Proposition 4]{gunduzyener09}:
\begin{proposition}\label{theorem:cf}
Consider an AWGN MWRC. Using the compress-forward coding strategy, the following equal rates are achievable:
\begin{equation}
R < \frac{1}{2(L-1)} \log \left( 1 + \frac{(L-1)PP_0}{1+(L-1)P + P_0} \right) \triangleq R_\text{CF}.
\end{equation}
\end{proposition}

\begin{remark}
It has been shown that the compress-forward coding strategy always achieves a higher equal rate than that achievable by the amplify-forward coding strategy~\cite[Remark 2]{gunduzyener09}. Hence, amplify-forward will not be considered in this paper.
\end{remark}

Although the aforementioned rates were derived for the restricted AWGN MWRC, they are equally applicable for the unrestricted AWGN MWRC considered in this paper.

\subsection{Functional-Decode-Forward}\label{section:fdf}

In this section, we propose a functional-decode-forward coding strategy for the AWGN MWRC. This strategy is based on our previous work on functional-decode-forward for the binary MWRC~\cite{ongjohnsonkellett10cl}, where the uplink transmission is split into $(L-1)$ blocks of $n$ channel uses, and the node pairs $(1,2), (2,3), \dotsc, (L-1,L)$ transmit using binary linear codes in the respective blocks. For the AWGN MWRC, we propose the following modifications for the uplink transmission:
\begin{enumerate}
\item We consider multiple message tuples, $(W_1^{(m)}, W_2^{(m)}, \dotsc, W_L^{(m)})$ for $m \in \{1,2,\dotsc\}$.
\item Instead of having fixed user pairs to transmit for each message tuple (as for the binary case), we \emph{rotate} the transmission scheme such that the pair of users that transmit in block $l \in \{1,2,\dotsc, L-1\}$ for message tuple $m \in \{1,2,\dotsc,\}$ are users $(l + m -2 \mod L) + 1$ and $(l + m - 1 \mod L) + 1$.
\item We use lattice codes for the uplink transmission.
\end{enumerate}
An example of the transmission scheme for $L=5$ for the first five message tuples is shown in Fig.~\ref{fig:scheme} of Appendix~\ref{appendix:example}.

Define $[x]^+ \triangleq \max \{ x, 0\}$.
We now show that functional-decode-forward achieves the following rate:
\begin{theorem}\label{theorem:ifdf}
Consider an AWGN MWRC. Using the functional-decode-forward coding strategy, the following equal rates are achievable:
\begin{align}
R < \min \Bigg\{ &\left[  \frac{1}{2(L-1)}\log \left( \frac{1}{2} +\frac{L}{2}P  \right) \right]^+, \nonumber \\ &  \frac{1}{2(L-1)} \log ( 1 + P_0) \Bigg\} \triangleq R_\text{FDF}. \label{eq:ifdf}
\end{align}
\end{theorem}

\begin{proof}
Consider $n$ uses of an AWGN multiple-access channel $Y_0 = X_1 + X_2 + N_0$, where $X_1$ and $X_2$ are the inputs with power  $\sum_{t=1}^nE[X_1^2[t]]/n = \sum_{t=1}^nE[X_2^2[t]]/n = P'$ and $N_0$ is the independent zero-mean  Gaussian noise with $E[N_0^2] =1$. Let $\mathcal{C}_\text{lattice}$ be an $n$-dimensional lattice code, where $|\mathcal{C}_\text{lattice}|=2^{nR}$. If the transmitters send $\boldsymbol{X}_i = \boldsymbol{V}_i + \boldsymbol{d}_i \mod \Lambda$, for $i \in \{1,2\}$, where $\boldsymbol{V}_i \in \mathcal{C}_\text{lattice}$ are lattice codewords, $\boldsymbol{d}_i$ are independent but fixed dither vectors, and $\!\!\!\!\mod \Lambda$ is the modulo-lattice operation~\cite{erezzamir04}, then the receiver can reliably decode $\boldsymbol{V}_{1,2} \triangleq \boldsymbol{V}_1(W_1) + \boldsymbol{V}_{2}(W_{2}) \mod \Lambda$ from $\boldsymbol{Y}_0$ if~\cite{nazergastpar08eu,erezzamir08israel,wilsonnarayananpfisersprintson10,namchunglee09,nazergastpar11}
\begin{equation}
R < \left[\frac{1}{2} \log \left( \frac{1}{2} + P' \right) \right]^+. \label{eq:fdf-2-users}
\end{equation}

\indent\indent{\em Uplink:}
Now, for the AWGN MWRC, each user $i$ bijectively maps its message $W_i^{(m)} \in \{1,2,\dotsc, 2^{nR}\}$ to a lattice codeword $\boldsymbol{V}_i(W_i^{(m)}) \in \mathcal{C}_\text{lattice}$. For message tuple $m$, in block $l$, nodes $(l + m -2 \mod L) + 1$ and $(l + m - 1 \mod L) + 1$ transmit, and all the other nodes do not transmit, i.e.,
\begin{equation}
\boldsymbol{X}_{i+1}(W_{i+1}^{(m)}) = \begin{cases}
\boldsymbol{V}_{i+1}(W_{i+1}^{(m)}) + \boldsymbol{d}_{i+1} \mod \Lambda, \\
\hfill \text{if } i= l + m -2 \mod L\\
\hfill \text{or } i= l + m - 1 \mod L\\
 \boldsymbol{0}, \quad\quad\quad\quad\quad\quad\quad\quad\quad\quad\quad\quad  \text{otherwise},
\end{cases}\label{eq:fdf-scheme-2}
\end{equation}
where $\boldsymbol{0}$ is the all-zero vector.

The aforementioned transmission scheme repeats itself after every $L$ message tuples. Consider a window of $L$ message tuples, e.g., from the first tuple to the $L$-th tuple. As there are $L-1$ blocks of transmission for each message tuple, there are all together $L(L-1)$ blocks of transmission. Since we cycle the transmission scheme according to \eqref{eq:fdf-scheme-2}, each node transmits in one block for two of the $L$ message tuples, and transmits in two blocks for the other $(L-2)$ of the $L$ message tuples. For the example of $L=5$, see Appendix~\ref{appendix:example}. If each user transmits with $\sum_{\text{one block}}E[X_i^2[\cdot]]/n = \frac{L(L-1)}{2 + 2(L-2)}P = \frac{LP}{2}$ when active, the average transmitted power is then $\sum_{t=1}^{nL(L-1)}E[X_i^2[t]]/{(nL(L-1))}=P$. So, we choose $P' = \frac{L}{2}P$. Note that had we used the same scheme for the binary MWRC, i.e., fixing the pairs $(1,2), (2,3), \dotsc, (L-1,L)$ to transmit for every message tuple, then each user could only transmit at the average power of $\frac{L-1}{2}P$, because all users must transmit at the same power using the same lattice code, and nodes $2,3,\dotsc,L-1$ each transmit in two of the $L-1$ blocks for each message tuple.

For any message tuple $t$, the user pairs that transmit in the $(L-1)$ blocks are $\{ (1,2), (2,3),$ $(3,4), \dotsc, (L-1,L), (L,1)\}$ except for the pair $([t + L-2 \mod L] + 1, [t+L-1 \mod L] + 1)$. In each block where a pair of users, say users $i$ and $j$, transmit, the relay decodes the function $\boldsymbol{V}_{i,j} \triangleq \boldsymbol{V}_i(W_i) + \boldsymbol{V}_{j}(W_{j}) \mod \Lambda$. From \eqref{eq:fdf-2-users}, we know that if
\begin{equation}
R < \left[\frac{1}{2(L-1)} \log \left( \frac{1}{2} + \frac{L}{2}P \right) \right]^+,\label{eq:fdf-uplink}
\end{equation}
then the relay can reliably decode the function $\boldsymbol{V}_{i,j}$ in each block. The factor $\frac{1}{L-1}$ takes into account that there are $(L-1)$ blocks for each message, and $P'=\frac{L}{2}P$ is the transmitted power of each active user in each block.

\indent\indent{\em Downlink:}
Note that each function $\boldsymbol{V}_{i,j} \in \mathcal{C}_\text{lattice}$, where $|\mathcal{C}_\text{lattice}| = 2^{nR}$. For a message tuple, let the cascade of the functions that the relay decodes in all the $(L-1)$ blocks be $\mathbb{W}$, which can be indexed by $U \in \{1,2,\dotsc, 2^{n(L-1)R}\}$.
On the downlink, for each message tuple, the relay broadcasts the corresponding $U$ to all the users. As each downlink channel is a point-to-point AWGN channel, each user can reliably decode $U$ if
\begin{equation}
(L-1)R < \frac{1}{2} \log (1 + P_0). \label{eq:fdf-downlink}
\end{equation}

\indent\indent{\em Decoding of Other Users' Messages:}
Now, consider the first message tuple, $m=1$. Suppose that each user has decoded $U$ sent by the relay. Since the mapping from $\mathbb{W}$ to $U$ is bijective, each user can recover $\mathbb{W} = (\boldsymbol{V}_{1,2}, \boldsymbol{V}_{2,3}, \dotsc, \boldsymbol{V}_{L-1,L})$. From $\mathbb{W}$ and the user's own lattice codeword $\boldsymbol{V}_i(W_i)$, the user can recover the codewords of all the other users~\cite{ongjohnsonkellett10cl}.
Since the mapping from $W_i$ to $\boldsymbol{V}_i(W_i)$ is also bijective for all $i$, every user can recover all other users' messages. It can be shown that for every message tuple, the aforementioned operation can be performed because the user pairs that transmit are $\{ (1,2), (2,3),$ $(3,4), \dotsc, (L-1,L), (L,1)\}$ less one.

Combining \eqref{eq:fdf-uplink} and \eqref{eq:fdf-downlink}, we get Theorem~\ref{theorem:ifdf}.
\end{proof}

\begin{remark}
Note that the strategy proposed here is different from the strategy described by G\"und\"uz et al.~\cite[Section IV. B.]{gunduzyener09}, which also uses lattice codes, where there is more than one cluster with two users in each cluster, and only the two users in each cluster exchange messages. When only two users exchange information, the previously proposed functional-decode-forward strategy for the TWRC~\cite{wilsonnarayananpfisersprintson10,namchunglee09} can be used. In the MWRC considered in this paper, there is only one cluster with $L$ users and all users engage in full data exchange.
\end{remark}

\begin{remark}
Setting $L=2$ for Theorem~\ref{theorem:ifdf} we recover the result for the TWRC~\cite{wilsonnarayananpfisersprintson10}.
\end{remark}

\begin{remark}
Note that the functional-decode-forward strategy relies on the condition that each user can form $L-1$ linearly independent equations consisting of the other users' messages (after removing its own message). This suggests $L-1$ blocks of transmission on the uplink. In each block, we are free to choose the set of users that transmit, as long as the linear-independence condition can be met for each user. On the one hand, if $K>2$ users transmit simultaneously in each block (using lattice codes), for the relay to decode the modulo-sum of the $K$ codewords, we impose $R < \frac{1}{2(L-1)} \log \left( \frac{1}{K} + P' \right)$~\cite{erezzamir08israel}. Furthermore, the more blocks in which a user transmits, the lower its transmitted power $P'$ to maintain an average of $P$. On the other hand, if only one user transmits in one of the blocks, after removing its codewords, the user will not get $L-1$ linearly independent equations from the relay. So, the minimum number of simultaneous transmitted codewords in each block is two if we fix the number of blocks to be $L-1$. In fact, if we have $L$ blocks and let one user transmit in each block, we get the complete-decode-forward rate [c.f. \eqref{eq:cdf}].
\end{remark}

\begin{remark}
Time-division multiple-access (TDMA) is in general sub-optimal in terms of spectral efficiency. However, one can show that in the symmetrical AWGN multiple-access channel, where all transmitters have the same power constraint and transmit at the same rate, TDMA is optimal (note that simultaneous transmission is always optimal in the multiple-access channel). We indeed have this configuration for the uplink of the AWGN MWRC. 
It is worth noting that functional-decode-forward, which uses TDMA on the uplink, achieves an uplink {\em pre-log factor} of $\frac{1}{2(L-1)}$ [see the first term on the RHS of \eqref{eq:ifdf}], which is higher than the pre-log factor $\frac{1}{2L}$ of the complete-decode-forward [see the first term on the RHS of \eqref{eq:cdf}], which uses simultaneous transmission on the uplink. In fact, functional-decode-forward achieves the pre-log factor of that of the uplink upper bound.
\end{remark}

\section{The Equal-Rate Capacity}\label{section:capacity}

We first show that complete-decode-forward and functional-decode-forward can achieve the equal-rate capacity upper bound. Let $R_\text{UB}' \triangleq \frac{1}{2(L-1)}\log(1+P_0)$, i.e., the last term on the RHS of \eqref{eq:upper-bound}. Note that $R_\text{UB} \leq R_\text{UB}'$. So, any coding strategy that achieves $R_\text{UB}'$ also achieves $C$.  

 Next, we show a sufficient condition for the rate $R_\text{UB}'$ to be achievable.
\begin{theorem} \label{theorem:capacity-general}
Consider an AWGN MWRC.  If
\begin{equation}
P_0 \leq \max \left\{ \frac{LP -1}{2}, (1 + LP) ^{\left( \frac{L-1}{L} \right)} -1 \right\},
\end{equation}
then the equal-rate capacity is
\begin{equation}
C = \frac{1}{2(L-1)} \log (1 + P_0).
\end{equation}
\end{theorem}

\begin{proof}[ Proof of Theorem~\ref{theorem:capacity-general}]
Note that $L \geq 2$.
If $P_0 \leq \frac{LP -1}{2}$, then $\frac{1}{2} + \frac{L}{2}P \geq 1 + P_0$, and hence $R_\text{FDF} = \frac{1}{2(L-1)} \log (1 + P_0) = C$. In addition, if $P_0 \leq (1 + LP) ^{\left( \frac{L-1}{L} \right)} -1$, then $\frac{1}{2(L-1)} \log ( 1 + P_0) \leq \frac{1}{2L} \log ( 1 + LP)$, and $R_\text{CDF} = \frac{1}{2(L-1)} \log (1 + P_0) = C$.
\end{proof}

\begin{remark} \label{remark:cdf-fdf-ub}
For any $L \geq 2$, we have that
\begin{subequations}
\begin{align}
&\max \left\{\frac{1}{2L} \log ( 1 + LP),\left[  \frac{1}{2(L-1)}\log \left( \frac{1}{2} +\frac{L}{2}P  \right) \right]^+ \right\} \nonumber \\ &< \frac{1}{2(L-1)} \log ( 1 + (L-1)P) \\ & \leq \frac{1}{2\ell} \log ( 1 + \ell P ) \leq \frac{1}{2\ell} \log ( 1 + \ell^2 P ), \label{eq:ub-of-ub}
\end{align}
\end{subequations}
for all  $\ell \in \{1,2,\dotsc,L-1\}$, because $\frac{1}{2\ell} \log ( 1 + \ell P )$ is a monotonically decreasing function of $\ell$. This means the first terms on the RHS of \eqref{eq:cdf} and \eqref{eq:ifdf} are strictly less than that of \eqref{eq:upper-bound}. So, complete-decode-forward or functional-decode-forward can achieve the equal-rate capacity upper bound only when $R_\text{UB} = R_\text{UB}'$, and when $R_\text{UB}'$ is achievable.
\end{remark}

Now, we investigate two special cases of equal transmitted power, and of scaling of the relay's transmitted power with the number of users.

\subsection{Equal Transmitted Power for Relay and Users}\label{cap}

In this section, we consider the case where the transmitted power of all users and the relay is equal, i.e., $P_0=P$. This means the SNR is $P$ for all nodes.\footnote{The case of equal transmitted power for $L=2$ (two users) was also considered by Wilson et al.~\cite{wilsonnarayananpfisersprintson10}.} Under this condition, the equal-rate capacity upper bound simplifies to $R_\text{UB} = \frac{1}{2(L-1)} \log ( 1 + P)$.

\begin{remark}
Note that $L$ is an integer greater than one. For any fixed $P$, as $L \rightarrow \infty$, $R_\text{UB} \rightarrow 0$. Any coding strategy that achieves a non-zero equal rate can be said to \emph{approach} the capacity upper bound (in an absolute sense) as $L$ increases. So, in this paper, we consider non-trivial cases of fixed $L$ and increasing $P$.
\end{remark}

 We have the following capacity result:
\begin{theorem}\label{theorem:cap}
Consider the AWGN MWRC with $P_0 = P$.
\begin{enumerate}
\item For $L \geq 3$: The equal-rate capacity is
\begin{equation}
C = \frac{1}{2(L-1)}\log ( 1 + P). \label{eq:cr-cap}
\end{equation}
\begin{itemize}
\item If $0 < P \leq 1$, the capacity is achievable by complete-decode-forward.
\item If $P \geq 1$, the capacity is achievable by functional-decode-forward.
\end{itemize}
\item For $L =2$:
\begin{equation}
R_\text{FDF} = \left[ \frac{1}{2}\log \left( \frac{1}{2} + P \right) \right]^+  \geq C - \epsilon(P),
\end{equation}
where $\epsilon(P) = \min \left\{ \frac{1}{2}, \frac{1}{2(2P+1) \ln 2} \right\}$.
\end{enumerate}
\end{theorem}

\begin{remark}
Note that  $\lim_{P \rightarrow \infty} \epsilon(P) = 0$ and $\lim_{P \rightarrow \infty} \frac{\epsilon(P)}{R_\text{UB}}= 0$. So, functional-decode-forward achieves the equal-rate capacity asymptotically as $P$ increases in an absolute sense as well as in a normalized (to the upper bound) sense.
\end{remark}

\begin{remark}
In complete-decode-forward and functional-decode-forward, each user's transmitted signals depend only on the user's messages So, $R_\text{CDF}$ and $R_\text{FDF}$ are also achievable on the restricted AWGN MWRC. Since the equal-rate capacity of the restricted AWGN MWRC must be upper bounded by $C$, we incidentally obtain the equal-rate capacity of the restricted  AWGN MWRC for $L\geq 3$ and $P_0 = P$.
\end{remark}

\begin{remark}
While the observation that functional-decode-forward achieves the capacity asymptotically for $L=2$ was made previously~\cite{wilsonnarayananpfisersprintson10,namchunglee09},  Theorem~\ref{theorem:cap}  provides an explicit upper bound on the gap.
\end{remark}

\begin{proof}[Proof of Theorem~\ref{theorem:cap}]
First consider $L \geq 3$. When $P \geq 1$, we have $\frac{L}{2}P \geq \frac{1}{2} + P$, and hence $\frac{1}{2} + \frac{L}{2}P \geq 1 + P$. So, $R_\text{FDF} = R_\text{UB} = C$.

From the proof of Theorem~\ref{theorem:capacity-general}, we know that the condition under which $R_\text{CDF} = R_\text{UB}$ is when $P \leq (1+LP)^{\frac{L-1}{L}} -1$, or equivalently, $\left(  \frac{1+LP}{1+P} \right)  ^{L-1} \geq 1+P$. Defining $\alpha(L,P) = \left(  \frac{1+LP}{1+P} \right)  ^{L-1}$ and $\beta(P) = 1+P$, we now show that if $0 < P \leq 1$, then $\alpha(L,P) \geq \beta(P)$. Note that $\alpha(L,0)=\beta(0)=1$, and $\frac{d}{dP}\beta(P) =1$ for  all $P$. In addition,
\begin{align}
\frac{d}{dP} \alpha(L,P) &= (L-1)^2 (1+P)^{-L} (1+LP)^{L-2} > 0\\
\frac{d^2}{dP^2} \alpha(L,P) &= (L-1)^2L(1+P)^{-L-1}(1+LP)^{L-3} \nonumber \\ &\quad \times (L-3-2P).
\end{align}
Recall that $L \geq 3$. We have $\left. \frac{d}{dP} \alpha(L,P) \right\vert_{P=0} > 1$, $\frac{d^2}{dP^2} \alpha(L,P)$ decreases as $P$ increases, and $\frac{d^2}{dP^2} \alpha(L,P) < 0$ when $P>\frac{L-3}{2}$. So, there exists a point $P^*(L)>0$ where $\alpha(L,P) \geq \beta(P)$ for $P \leq P^*(L)$, and $\alpha(L,P) < \beta(P)$ for $P > P^*(L)$. Now, fixing $P=1$, we have 
\begin{subequations}
\begin{align}
(1+L)/2 &\geq 2\\
\left(\frac{1+LP}{1+P}\right)^{L-1} &\geq 1+P\\
\alpha(L,P) &\geq \beta(P).
\end{align}
\end{subequations}
This means $P=1$ falls into the region in which $\alpha(L,P) \geq \beta(P)$, meaning $P^*(L) \geq 1$. So, $R_\text{CDF} = R_\text{UB}$ for $L \geq 3$ and $0 < P \leq 1$.

Next, for $L = 2$, we have $R_\text{FDF} = \frac{1}{2}\log \left( \frac{1}{2} + P \right)$ as the first term is smaller than the second term in the RHS of \eqref{eq:ifdf}. Note that $\frac{d}{dx} \log x = \frac{1}{x \ln 2}$ and $\frac{d^2}{dx^2} \log x = - \frac{1}{x^2 \ln 2} < 0$.
So,
\begin{subequations}
\begin{align}
\log \left(x+ \delta \right) &< \log x + \left. \frac{d}{dy} \log y \right\vert_{y=x} \left( \left(x+\delta\right) - x \right)\\
&= \log x + \frac{\delta}{x \ln 2}. \label{eq:log-gap}
\end{align}
\end{subequations}
Hence,
\begin{subequations}
\begin{align}
C &\leq R_\text{UB}\\
&< \frac{1}{2} \log \left( \frac{1}{2} + P \right) + \frac{1}{2} \frac{\frac{1}{2}}{\left(P + \frac{1}{2} \right) \ln 2} \\
&= R_\text{FDF} + \frac{1}{2\left(2P+1 \right) \ln 2}. \label{eq:fdf-ub-gap}
\end{align}
\end{subequations}

Furthermore, $ C \leq R_\text{UB} = \frac{1}{2} \log ( 1 + P ) = \frac{1}{2} \log \left(2 \left( \frac{1}{2} + \frac{P}{2} \right) \right) =\frac{1}{2} \log  \left( \frac{1}{2} + \frac{P}{2} \right)  + \frac{1}{2} < R_\text{FDF} + \frac{1}{2}$.

\end{proof}

In addition to the capacity results in Theorem~\ref{theorem:cap}, we have the following theorem:

\begin{theorem} \label{theorem:common-addition}
Consider the AWGN MWRC with $P_0=P$ and $L \geq 3$. We have the following:
\begin{enumerate}
\item Complete-decode-forward:
\begin{enumerate}
\item $R_\text{CDF} < C$, if $P \geq L^{L-1} -1$;
\item $R_\text{CDF} = C  -\mathcal{O}(\log P)$, as $P \rightarrow \infty$.
\end{enumerate}
\item Compress-Forward:
\begin{enumerate}
\item $R_\text{CF} < C$, for all $P$;
\item $R_\text{CF} \rightarrow  C - \frac{1}{2(L-1)}  \log \left( 1 + \frac{1}{(L-1)} \right)$, as $P \rightarrow \infty$.
\end{enumerate}
\end{enumerate}
\end{theorem}

\begin{proof}
The proof is presented in Appendix~\ref{appendix:common-addition}.
\end{proof}

The aforementioned theorem implies that when there are three or more users, of the coding strategies studied in this paper, only functional-decode-forward achieves the equal-rate capacity in the high SNR regime.

\subsection{Transmitted Power for Relay Scales with $L$}

In this section, we consider the case where the transmitted power of all users is equal but the transmitted power of the relay scales with the number of users, i.e., $P_0=LP$. This can be used to model networks in which a line-powered base station (the relay) serves multiple mobile users, and we can increase the transmitted power of the base station to accommodate the addition of new users. In this case of scaling the relay's transmitted power, the SNR at the relay is $P$, but the SNR at the user is $LP$. Unless otherwise stated, we take SNR to mean the SNR at the relay, which is $P$.

For the case of scaling power, we have $R_\text{CDF} = \frac{1}{2L} \log ( 1 + LP)$, and $R_\text{FDF} = \left[\frac{1}{2(L-1)} \log\left( \frac{1}{2} + \frac{L}{2}P\right) \right]^+$. It follows from Remark~\ref{remark:cdf-fdf-ub} that $\max\{R_\text{CDF},R_\text{FDF}\} < R_\text{UB}$. It is also possible to show that $R_\text{CF} < R_\text{UB}$ (see Appendix~\ref{sec:cf-sub}). Therefore, none of the coding strategies discussed in this paper can achieve $R_\text{UB}$. However, we will investigate how close a rate we can achieve compared to the upper bound $R_\text{UB}' = \frac{1}{2(L-1)}\log(1+LP)$. Since $R_\text{UB}' \rightarrow 0$, as $L \rightarrow \infty$, we will consider fixed $L$ and increasing $P$.

\begin{theorem} \label{theorem:scaling}
Consider the AWGN MWRC with $P_0=LP$. We have the following:
\begin{enumerate}
\item Functional-decode-forward:
\begin{enumerate}
\item $R_\text{FDF} > C - \frac{1}{2(L-1)\ln 2}$, for all $P$.
\end{enumerate}
\item Complete-decode-forward:
\begin{enumerate}
\item $R_\text{CDF} < R_\text{UB}$, for all $P$;
\item $R_\text{CDF} < C - \mathcal{O}(\log P)$,  as $P \rightarrow \infty$.
\end{enumerate}
\item Compress-forward:
\begin{enumerate}
\item $R_\text{CF} \geq C - \frac{1}{2(L-1)} \log \left( \frac{1 + (2L-1)P}{1 + (L-1)P} \right)$;
\item $R_\text{CF} \geq C - \frac{1}{2(L-1)} \log\left( 1 + \frac{L}{L-1}\right)$, as $P \rightarrow \infty$.
\end{enumerate}
\end{enumerate}
\end{theorem}
\begin{proof}
The proof is presented in Appendix~\ref{appendix:scaling}.
\end{proof}

\begin{remark}
It has been shown that compress-forward achieves rates within a constant $\frac{1}{2(L-1)}$ bits of the equal-rate capacity of the restricted AWGN MWRC~\cite[Theorem 1]{gunduzyener09}. In this paper, we further show that this coding strategy can also achieve within a constant fraction of a bit of the equal-rate capacity of the unrestricted AWGN MWRC when $P_0 = LP$.
\end{remark}

Furthermore, we have the following comparison at high SNR:
\begin{corollary}\label{theorem:ifdf-grater-than-cf}
Consider the AWGN MWRC with $P_0=LP$. If $P > \frac{L-1 + \sqrt{ (L-1)^2 + 4L}}{2L}$, then $R_\text{FDF} > R_\text{CF}$.
\end{corollary}

\begin{proof}
If $P > \frac{L-1 + \sqrt{ (L-1)^2 + 4L}}{2L}$, then $LP^2 - (L-1)P -1 > 0$, meaning that $\frac{1}{2(L-1)} \log \left( \frac{1}{2} + \frac{LP}{2} \right) > \frac{1}{2(L-1)} \log \left( 1 + \frac{L(L-1)P^2}{1 + (L-1)P + LP} \right)$, and thus $R_\text{FDF} > R_\text{CF}$.
\end{proof}

From 1(a) and 2(b) of Theorem~\ref{theorem:scaling}, we know that functional-decode-forward achieves rates within a constant fraction of a bit of the equal-rate capacity, and the complete-decode-forward rates are bounded away from the capacity at high SNR. From Theorem~\ref{theorem:ifdf-grater-than-cf}, functional-decode-forward achieves higher rates than compress-forward at high SNR. Hence, in the high SNR regime, functional-decode-forward outperforms the other two schemes.

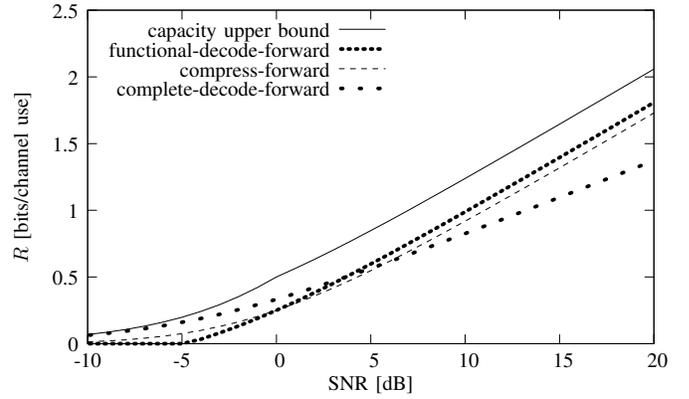
\begin{figure}[t]
\centering
\resizebox{\linewidth}{!}{ 
\begin{picture}(0,0)%
\includegraphics{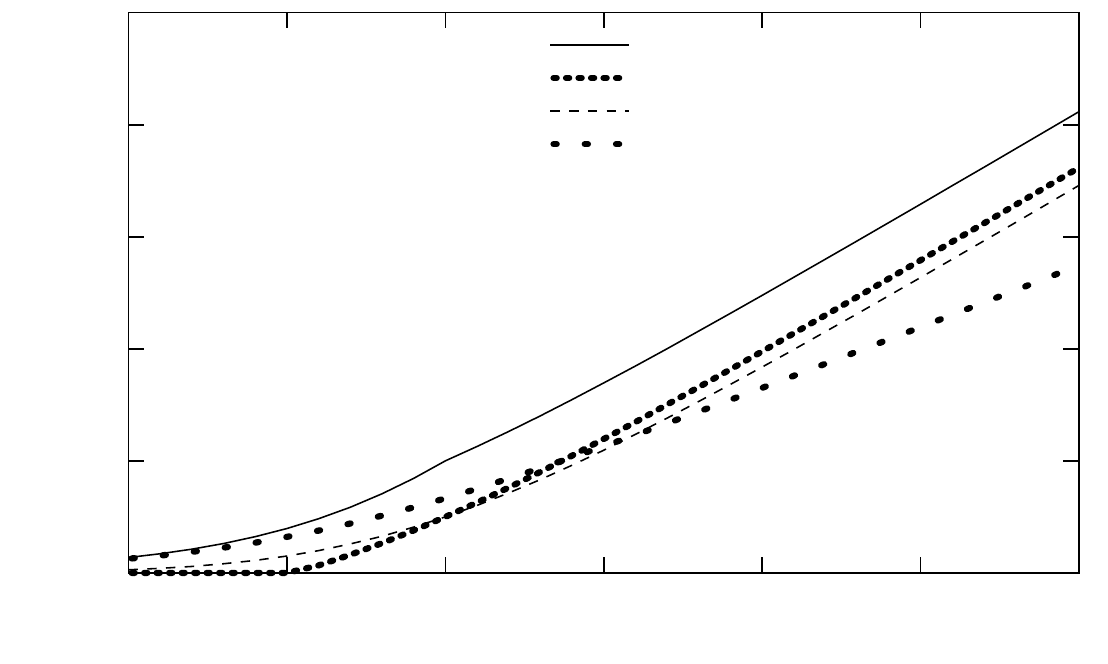}%
\end{picture}%
\setlength{\unitlength}{3947sp}%
\begingroup\makeatletter\ifx\SetFigFont\undefined%
\gdef\SetFigFont#1#2#3#4#5{%
  \reset@font\fontsize{#1}{#2pt}%
  \fontfamily{#3}\fontseries{#4}\fontshape{#5}%
  \selectfont}%
\fi\endgroup%
\begin{picture}(5309,3181)(1196,-3995)
\put(3763,-1088){\makebox(0,0)[rb]{\smash{{\SetFigFont{10}{12.0}{\familydefault}{\mddefault}{\updefault}capacity upper bound}}}}
\put(3763,-1247){\makebox(0,0)[rb]{\smash{{\SetFigFont{10}{12.0}{\familydefault}{\mddefault}{\updefault}functional-decode-forward}}}}
\put(3763,-1406){\makebox(0,0)[rb]{\smash{{\SetFigFont{10}{12.0}{\familydefault}{\mddefault}{\updefault}compress-forward}}}}
\put(3763,-1565){\makebox(0,0)[rb]{\smash{{\SetFigFont{10}{12.0}{\familydefault}{\mddefault}{\updefault}complete-decode-forward}}}}
\put(1738,-3623){\makebox(0,0)[rb]{\smash{{\SetFigFont{10}{12.0}{\familydefault}{\mddefault}{\updefault} 0}}}}
\put(1738,-3085){\makebox(0,0)[rb]{\smash{{\SetFigFont{10}{12.0}{\familydefault}{\mddefault}{\updefault} 0.5}}}}
\put(1738,-2547){\makebox(0,0)[rb]{\smash{{\SetFigFont{10}{12.0}{\familydefault}{\mddefault}{\updefault} 1}}}}
\put(1738,-2010){\makebox(0,0)[rb]{\smash{{\SetFigFont{10}{12.0}{\familydefault}{\mddefault}{\updefault} 1.5}}}}
\put(1738,-1472){\makebox(0,0)[rb]{\smash{{\SetFigFont{10}{12.0}{\familydefault}{\mddefault}{\updefault} 2}}}}
\put(1738,-934){\makebox(0,0)[rb]{\smash{{\SetFigFont{10}{12.0}{\familydefault}{\mddefault}{\updefault} 2.5}}}}
\put(1813,-3748){\makebox(0,0)[b]{\smash{{\SetFigFont{10}{12.0}{\familydefault}{\mddefault}{\updefault}-10}}}}
\put(2573,-3748){\makebox(0,0)[b]{\smash{{\SetFigFont{10}{12.0}{\familydefault}{\mddefault}{\updefault}-5}}}}
\put(3334,-3748){\makebox(0,0)[b]{\smash{{\SetFigFont{10}{12.0}{\familydefault}{\mddefault}{\updefault} 0}}}}
\put(4094,-3748){\makebox(0,0)[b]{\smash{{\SetFigFont{10}{12.0}{\familydefault}{\mddefault}{\updefault} 5}}}}
\put(4854,-3748){\makebox(0,0)[b]{\smash{{\SetFigFont{10}{12.0}{\familydefault}{\mddefault}{\updefault} 10}}}}
\put(5615,-3748){\makebox(0,0)[b]{\smash{{\SetFigFont{10}{12.0}{\familydefault}{\mddefault}{\updefault} 15}}}}
\put(6375,-3748){\makebox(0,0)[b]{\smash{{\SetFigFont{10}{12.0}{\familydefault}{\mddefault}{\updefault} 20}}}}
\put(1331,-2217){\rotatebox{90.0}{\makebox(0,0)[b]{\smash{{\SetFigFont{10}{12.0}{\familydefault}{\mddefault}{\updefault}$R$ [bits/channel use]}}}}}
\put(4094,-3935){\makebox(0,0)[b]{\smash{{\SetFigFont{10}{12.0}{\familydefault}{\mddefault}{\updefault}$\text{SNR}$ [dB]}}}}
\end{picture}%
}
\caption{Comparing different coding strategies and the capacity upper bound ($P_0=LP$, $L=3$, varying $\text{SNR}=P$)}
\label{fig:pfdf-3}
\end{figure}

\begin{figure*}
\begin{center}
\begin{small}
  \begin{tabular}{| c || c | c| c | c || c | c| c | c || c | c| c | c || c | c| c | c || c| c| c | c |}
\hline
 Message tuple ($m$)& \multicolumn{4}{c ||}{$1$} & \multicolumn{4}{c ||}{$2$}& \multicolumn{4}{c ||}{$3$}& \multicolumn{4}{c ||}{$4$} & \multicolumn{4}{c |}{$5$}\\
\hline
 Block ($l$) & $1$ & $2$ & $3$ & $4$ & $1$ & $2$ & $3$ & $4$  & $1$ & $2$ & $3$ & $4$ & $1$ & $2$ & $3$ & $4$ & $1$ & $2$ & $3$ & $4$ \\
\hline
\hline
Node & \multicolumn{20}{c |}{transmission}\\
\hline
1 & $\blacksquare$ & $\square$ & $\square$ & $\square$ & $\square$ & $\square$ & $\square$ & $\blacksquare$ & $\square$ & $\square$ & $\blacksquare$ &  $\blacksquare$ & $\square$ &  $\blacksquare$ & $\blacksquare$ & $\square$ &  $\blacksquare$ &  $\blacksquare$ &   $\square$ & $\square$\\
2 & $\blacksquare$ & $\blacksquare$ & $\square$ & $\square$ & $\blacksquare$  & $\square$ & $\square$ & $\square$ & $\square$ & $\square$ & $\square$ & $\blacksquare$ & $\square$ &  $\square$ & $\blacksquare$ & $\blacksquare$ &  $\square$ &  $\blacksquare$ &  $\blacksquare$ & $\square$\\
3 & $\square$ & $\blacksquare$ & $\blacksquare$ & $\square$ & $\blacksquare$ & $\blacksquare$ & $\square$ & $\square$ & $\blacksquare$ & $\square$ & $\square$ & $\square$ & $\square$ &  $\square$ &  $\square$ & $\blacksquare$ &  $\square$ &  $\square$ &  $\blacksquare$ & $\blacksquare$\\
4 & $\square$ & $\square$ & $\blacksquare$ & $\blacksquare$  & $\square$ & $\blacksquare$ & $\blacksquare$  &  $\square$ &  $\blacksquare$ & $\blacksquare$ & $\square$ &$\square$ & $\blacksquare$  & $\square$ &  $\square$ &  $\square$ & $\square$ &  $\square$ &    $\square$ & $\blacksquare$\\
5 & $\square$ & $\square$ & $\square$ & $\blacksquare$ &  $\square$ & $\square$ & $\blacksquare$ & $\blacksquare$ & $\square$ & $\blacksquare$ & $\blacksquare$ & $\square$ & $\blacksquare$ &  $\blacksquare$ & $\square$ &  $\square$ &  $\blacksquare$ &  $\square$ &  $\square$ & $\square$\\
\hline
\end{tabular}
\end{small}
\end{center}
\caption{Functional-decode-forward for five users}
\label{fig:scheme}
\end{figure*}

Fig.~\ref{fig:pfdf-3} compares the maximum achievable rates of different coding strategies and the capacity upper bound. We fix the number of users at $L=3$, and vary the SNR (defined as the SNR at the relay).
At low SNR, complete-decode-forward achieves rates higher than the other two coding strategies. As SNR increases, the performance of complete-decode-forward degrades, as predicted in Theorem~\ref{theorem:scaling}.
As expected, functional-decode-forward and compress-forward achieve rates within constant gaps from the capacity upper bound, and the former achieves a higher rate at high SNR. In this example, we see that at any SNR, either complete-decode-forward or functional-decode-forward outperforms compress-forward. This phenomenon is also observed for different $L$, and leads to the following conjecture:
\begin{conjecture}
Consider the AWGN MWRC with $P_0=LP$. For any $L$ and any $P$, $R_\text{CF} < \max \{ R_\text{CDF}, R_\text{FDF} \}$.
\end{conjecture}

\section{Conclusions}\label{section:conclusion}
We extended functional-decode-forward to the AWGN multiway relay channel (MWRC), and derived a new achievable equal-rate region for the channel. Combining this strategy with an existing strategy, complete-decode-forward, we obtained the capacity for the AWGN MWRC with three or more users when all nodes transmit at the same power. Thus the AWGN MWRC is one that, at the time of writing, requires more than one coding strategy to achieve the capacity. The capacity-achieving strategies do not utilize feedback from the channel to the users, which implies that feedback does not increase the equal-rate capacity of the AWGN MWRC when all nodes transmit at the same power.

For the case when the relay power scales with the number of users, the bottleneck of the network is no longer on the downlink, and we demonstrated that none of functional-decode-forward, complete-decode-forward, or compress-forward is able to achieve the capacity upper bound. However, functional-decode-forward and compress-forward achieve rates within a fixed number of bits of the capacity.

Numerical results suggest that for any $L$ and at any SNR, either complete-decode-forward or functional-decode-forward is able to outperform compress-forward. So, for AWGN MWRCs with equal transmitted power and where the relay's transmitted power scales with the number of users, among the four coding strategies, namely, functional-decode-forward, complete-decode-forward, compress-forward, and amplify-forward, as far as the transmission rate is concerned, complete-decode-forward and functional-decode-forward will suffice.

\appendices

\section{Functional-Decode-Forward for Five Users} \label{appendix:example}
The transmission scheme for $L=5$ for the first five message tuples is shown in Fig.~\ref{fig:scheme}. The symbol $\blacksquare$ means that the node (say node $i$) transmits using lattice codes $\boldsymbol{X}_{i}=\boldsymbol{V}_i(W_i) + \boldsymbol{d}_i \mod \Lambda$ with power $\sum_{\text{one block}} E[X_i^2[\cdot]]/n = \frac{5}{2}P$; while $\square$ means that it does not transmit, i.e., $\boldsymbol{X}_{i}= \boldsymbol{0}$.
Since, each node only transmits in a fraction of $\frac{8}{20}$ of the time, the average transmitted power is $P$.

\section{Proof of Theorem~\ref{theorem:common-addition}} \label{appendix:common-addition}

\subsection{Complete-Decode-Forward}
If $\frac{1}{2(L-1)} \log ( 1 + P) - \frac{1}{2L} \log (1 + LP) \geq 0$, then $R_\text{CDF} = \frac{1}{2L} \log (1 + LP)$.
Now,
\begin{subequations}
\begin{align}
&\frac{1}{2(L-1)} \log ( 1 + P) - \frac{1}{2L} \log (1 + LP) \nonumber \\  &= \frac{1}{2L(L-1)} \left( L \log(1+P) - (L-1) \log (1+LP) \right)\\
&=\frac{1}{2L(L-1)} \left( \log (1+P) - (L-1) \log \left( \frac{LP+1}{P+1} \right) \right) \label{eq:sufficient-condigion-positive-gap}\\
&> \frac{1}{2L(L-1)} \left( \log (1+P) - (L-1) \log L \right), \label{eq:l}
\end{align}
\end{subequations}
where the last inequality is because $1 < \left( \frac{LP+1}{P+1} \right)^{L-1}  < L^{L-1}$ for $L \geq 3$ and $P > 0$. So, if $P \geq L^{L-1}-1$, the RHS of \eqref{eq:sufficient-condigion-positive-gap} is strictly positive, meaning that  $R_\text{CDF} = \frac{1}{2L} \log (1 + LP)$, and $C - R_\text{CDF} > 0$.

For any fixed $L$, from \eqref{eq:sufficient-condigion-positive-gap}, as $P \rightarrow \infty$, we have $R_\text{CDF} = C - \frac{1}{2L(L-1)} \left( \log (1+P) - (L-1) \log \left( \frac{L+\frac{1}{P}}{1+\frac{1}{P}} \right) \right)\\
= C - \mathcal{O} (\log P)$.


\subsection{Compress-Forward}
\begin{subequations}
\begin{align}
R_\text{CF} &= \frac{1}{2(L-1)} \log \left( 1 + \frac{(L-1)P^2}{1+(L-1)P + P} \right) \\
&= \frac{1}{2(L-1)} \log( 1 + P) + \frac{1}{2(L-1)} \log \left( \frac{1 + (L-1)P}{1 + LP} \right)\\
&= C - \frac{1}{2(L-1)} \log \left( 1 +  \frac{P}{1 + (L-1)P} \right)\\
&= C - \frac{1}{2(L-1)} \log \left( 1 + \frac{1}{(L-1) + \frac{1}{P}} \right). \label{eq:assym-cf}
\end{align}
\end{subequations}

For any fixed $L$, as $P \rightarrow \infty$, we have $R_\text{CF} \rightarrow R_\text{UB} - \frac{1}{2(L-1)}  \log \left( 1 + \frac{1}{(L-1)} \right)$.

\section{Proof of Theorem~\ref{theorem:scaling}} \label{appendix:scaling}

\subsection{Functional-Decode-Forward}
When $P_0=LP$, $R_\text{FDF} = \frac{1}{2(L-1)} \log \left( \frac{1}{2} + \frac{L}{2}P \right)$.
Now,
\begin{subequations}
\begin{align}
C &\leq R_\text{UB}'= \frac{1}{2(L-1)} \log ( 1 + LP)\\
&= \frac{1}{2(L-1)} \log \left( \frac{1}{2} + \frac{L}{2}P + \frac{1}{2} + \frac{L}{2}P \right)\\
&< \frac{1}{2(L-1)} \left( \log \left( \frac{1}{2} + \frac{L}{2}P \right) + \frac{\frac{1}{2} + \frac{L}{2}P}{\left(\frac{1}{2} + \frac{L}{2}P \right) \ln 2 }  \right) \label{eq:log-gap-2}\\
&= R_\text{FDF} + \frac{1}{2(L-1) \ln 2},\label{eq:upper-bound-and-capacity-2}
\end{align}
\end{subequations}
where \eqref{eq:log-gap-2} follows from \eqref{eq:log-gap}. So, $R_\text{FDF} > C - \frac{1}{2(L-1) \ln 2}$.

\subsection{Complete-Decode-Forward}
When $P_0=LP$, $R_\text{CDF} = \frac{1}{2L} \log ( 1 + LP) <  \frac{1}{2(L-1)} \log ( 1 + LP)  = R_\text{UB}' \leq R_\text{UB}$.
Next,
\begin{subequations}
\begin{align}
R_\text{CDF} &= R_\text{UB}' - \left[ \frac{1}{2(L-1)} - \frac{1}{2L} \right]\log ( 1 + LP) \\
&< C + \frac{1}{2(L-1) \ln 2}  -\left[ \frac{1}{2(L-1)} - \frac{1}{2L} \right]\log ( 1 + LP), \label{eq:upper-bound-and-capacity}
\end{align}
\end{subequations}
where \eqref{eq:upper-bound-and-capacity} follows from \eqref{eq:upper-bound-and-capacity-2} and $R_\text{FDF} \leq C$.
So, for any fixed $L$, $R_\text{CDF} < C -\mathcal{O}(\log P)$.

\subsection{Compress-Forward} \label{sec:cf-sub}
\begin{subequations}
\begin{align}
R_\text{CF} &= \frac{1}{2(L-1)} \log \left( 1 + \frac{L(L-1)P^2}{1 + (L-1)P + LP} \right)\\
&= \frac{1}{2(L-1)} \log (1 + LP) \nonumber \\ &\quad + \frac{1}{2(L-1)} \log \left( \frac{1 + (L-1)P}{1 + (2L-1)P} \right)\\
&= R_\text{UB}' - \frac{1}{2(L-1)} \log \left( \frac{1 + (2L-1)P}{1 + (L-1)P} \right). \label{eq:cf-sub}
\end{align}
\end{subequations}
So, $R_\text{CF} \geq C - \frac{1}{2(L-1)} \log \left( \frac{1 + (2L-1)P}{1 + (L-1)P} \right)$, and $\lim_{P \rightarrow \infty} \frac{1}{2(L-1)} \log \left( \frac{1 + (2L-1)P}{1 + (L-1)P} \right) = \frac{1}{2(L-1)} \log\left( 1 + \frac{L}{L-1}\right)$.

We can also show that $R_\text{CF} < \frac{1}{2(L-1)} \log (1 + (L-1)P) \leq \frac{1}{2\ell} \log ( 1 + \ell^2 P )$, for all $\ell \in \{1,2,\dotsc, L-1\}$, where the last inequality follows from \eqref{eq:ub-of-ub}. So, $R_\text{CF} < R_\text{UB}$.



\begin{IEEEbiographynophoto}{Lawrence Ong}
(S'05--M'10) received the BEng (1st Hons) degree in electrical engineering from the National University of Singapore (NUS), Singapore, in 2001. He subsequently received the MPhil degree from the University of Cambridge, UK, in 2004 and the PhD degree from NUS in 2008.

He was with MobileOne, Singapore, as a system engineer from 2001 to 2002. He was a research fellow at NUS, from 2007 to 2008. From 2008 to 2012, he was a postdoctoral researcher at The University of Newcastle, Australia.
In 2012, he was awarded the Discovery Early Career Researcher Award (DECRA) by the Australian Research Council (ARC). He is currently a DECRA fellow at The University of Newcastle.
\end{IEEEbiographynophoto}

\begin{IEEEbiographynophoto}{Christopher M.\ Kellett}
(M'97--SM'10) received the Bachelor of Science in Electrical
Engineering and Mathematics from the University of California, Riverside
and the Master of Science and Doctor of Philosophy in Electrical
and Computer Engineering from the University of California, Santa Barbara.
He subsequently held research positions with the Centre Automatique et Syst\`emes
at \'Ecole des Mines de Paris, the Department of Electrical and Electronic
Engineering at the University of Melbourne, Australia, and the Hamilton Institute
at the National University of Ireland, Maynooth.  Since 2006 he has been
with the School of Electrical Engineering and Computer Science at the University of
Newcastle, Australia where he is currently an Australian Research Council
Future Fellow.  In 2012, Dr.~Kellett was awarded a  Humboldt Research Fellowship
funded by the Alexander von Humboldt Foundation, Germany.

\end{IEEEbiographynophoto}

\begin{IEEEbiographynophoto}{Sarah J.\ Johnson}
(M'04) received the B.E.\ (Hons) degree in electrical engineering in 2000, and PhD in 2004, both from the University of Newcastle, Australia.

She has held research positions with the Wireless Signal Processing
Program, National ICT Australia and the University of Newcastle. In
2011, Dr.\ Johnson was awarded a Future Fellowship by the Australian
Research Council (ARC). She is currently a future fellow at The
University of Newcastle.

\end{IEEEbiographynophoto}

\end{document}